\documentclass[12pt]{article}

\usepackage{fullpage}
\usepackage{times}
\usepackage{amsmath,amsfonts,amssymb,mathabx}
\usepackage{amsthm}
\usepackage{pdfpages}
\usepackage{float}
\usepackage{eso-pic}
\usepackage{ifthen}
\usepackage{pstricks}

\newtheorem{theorem}{Theorem}
\newtheorem{corollary}[theorem]{Corollary}
\newtheorem{lemma}[theorem]{Lemma}
\newtheorem{claim}[theorem]{Claim}

\newtheorem{proposition}[theorem]{Proposition}

{\theoremstyle{remark} }
{\theoremstyle{definition} \newtheorem{definition}[theorem]{Definition}}

\newenvironment{proofof}[1]{ {\sc Proof of #1.}\/}{\hfill\qedsymbol}

\renewcommand{\qedsymbol}{\ensuremath{\blacksquare}}

\newcommand{\remove}[1]{}
\newcommand{\suppress}[1]{}

\newcommand{\RR}{{\mathbb{R}}}

\DeclareMathOperator{\E}{E}

\newcommand{\eps}{\epsilon}
\newcommand{\veps}{\varepsilon}

\begin{document}

\title{Convergence of T\^{a}tonnement in Fisher Markets}

\author{
Noa Avigdor-Elgrabli\thanks{Yahoo! Labs Haifa,
MATAM, Haifa 31095, Israel.
Email: {\tt noaa@yahoo-inc.com}} \and
Yuval Rabani\thanks{The Rachel and Selim Benin School of Computer
Science and Engineering, The Hebrew University of Jerusalem,
Jerusalem 91904, Israel. Email: {\tt yrabani@cs.huji.ac.il}.
Research supported by Israel Science Foundation grant number
856-11 and by the Israeli Center of Excellence on Algorithms. Part
of this work was done while visiting Microsoft Research.}
\and Gala Yadgar\thanks{Computer Science Department,
Technion---Israel Institute of Technology, Haifa 32000, Israel.
Email: {\tt gala@cs.technion.ac.il}. This work was supported in
part by the Israel Ministry of Industry, Trade, and Labor under
the Net-HD Consortium, and by Microsoft.} }

\date{\today}

\setcounter{footnote}{3}
\maketitle

\begin{abstract}
Analyzing simple and natural price-adjustment processes that
converge to a market equilibrium is a fundamental question in
economics. Such an analysis may have implications in economic
theory, computational economics, and distributed systems.
T\^{a}tonnement, proposed by Walras in 1874, is a process
by which prices go up in response to excess demand, and
down in response to excess supply. This paper analyzes
the convergence of a time-discrete t\^{a}tonnement process, a
problem that recently attracted considerable attention of computer
scientists. We prove that the simple t\^{a}tonnement process that
we consider converges (efficiently) to equilibrium prices and
allocation in markets with nested CES-Leontief utilities, generalizing 
some of the previous convergence proofs for more restricted
types of utility functions.
\end{abstract}

\thispagestyle{empty}
\newpage
\setcounter{page}{1}

\section{Introduction}

\paragraph{Motivation.}
General equilibrium theory, a cornerstone of microeconomics,
deals with {\em markets} that consist of {\em agents} that are
endowed with {\em goods} (or money). Each agent wants to
maximize its {\em utility}. This may involve exchanging the initial
endowment of goods for goods that other agents hold. The exchange
of goods against other goods (or money) is governed by prices
that are set to guarantee a {\em market equilibrium} where
supply equals demand. In this context, an issue that was
raised already by the founders of this theory is that justifying
market equilibrium as a likely outcome of exchange requires
suggesting a plausible market dynamic that leads to equilibrium.
In other words, proving that some simple, natural,
decentralized price-adjustment process converges quickly to
a market equilibrium is a fundamental challenge of general
equilibrium theory. Indeed, Walras~\cite{Wal74} proposed
such a process which he named {\em t\^{a}tonnement}. In
this process, prices are adjusted upwards in response to
excess demand and downwards in response to excess supply.
(Notice that while buyers are expected to
react to price changes by maximizing their utility subject to
their budget constraints, the price changes themselves cannot
be justified based on the incentives of the sellers.)
In modern economic thought, t\^{a}tonnement is regarded as
a model for a centrally planned economy rather than a
decentralized market economy. In other words, it is considered
a method of computation rather than a spontaneous market process.
A scenario where this view of t\^{a}tonnement might
be useful arises in the context of distributed large-scale computer
systems, such as clouds. Artificial markets can be set to facilitiate
cooperative sharing of computing, storage, and communication
resources among selfish agents
(see, for example,~\cite{BAGS02,ZCY03,YFS13}).
Often, a service provider either controls the dispersed resources,
or at least regulates the peer-to-peer interaction among the
users, thus enforcing
adherence to the price adjustment protocol
(see, for instance,~\cite{BBST13,ZACLHDMWP13}).
On the other
hand, the buying agents are still free to pursue their goals
without having to formulate and disclose their utility functions
beyond their reactions to the changing prices. (One of the reasons
that this is important is that the conventional economic interpretation
of utility functions is that they are a convenient expression of
a preference order on the continuum of consumption baskets,
rather than a meaningful and conscious numerical evaluation of
each basket.)

\paragraph{Our results.}
We consider in this paper the following simple discrete
t\^{a}tonnement rule. For some fixed $\eps > 0$, for
every good, its price $p(t)$ at time $t=1,2,3,\dots$ is
given by
\begin{equation}\label{eq: linear tatonnement}
p(t) = (1 - \eps)\cdot p(t-1) + \eps\cdot f(t)\cdot p(t-1),
\end{equation}
where $f(t)$ is the total demand for this good assuming
the prices at time $t-1$ (the supply is always scaled to $1$).
We prove that this process converges to equilibrium
prices and allocation in various Eisenberg-Gale markets,
most notably markets with nested CES-Leontief utility
functions. We also note that the rate of convergence gives 
a polynomial time
algorithm for computing an approximate equilibrium---the
precise bounds are stated in Sections~\ref{sec: preliminary},
\ref{sec: nested}, and~\ref{sec: additional}.
Most previous results on the convergence of discrete versions of
t\^{a}tonnement apply to special cases of our results. Our
results apply to utility functions
for which previous work did not establish convergence of
t\^{a}tonnement. We note that resource
sharing in computer systems may require rather complicated
utility functions that could not be handled by previous work
(for instance, users may desire a combination of alternative
bundles of resources).

\paragraph{Markets.}
A Fisher market (see~\cite{Sca73}) is composed of a set of
agents and a set
of perfectly divisible goods, each of limited quantity. Each
agent is endowed with a budget and a
utility function on the possible baskets of goods. The
goal of an agent is to spend the budget to buy an optimal
basket of goods which maximizes the agent's utility subject
to the budget constraint. An equilibrium consists of an
assignment of prices to goods and an optimal purchase
of goods by each agent in which the market clears, i.e.,
the demand for each good equals to its supply.
Eisenberg-Gale markets~\cite{JV10} are a rather general
special case, where the utility functions are such that
equilibrium prices and allocation can be formulated as
the solutions to dual convex programs, originally due to
Eisenberg and Gale~\cite{EG59} and extended
in~\cite{Eis61,KV02,CV04,JVY05}.
Fisher markets are a special case of the more general
Walrasian model~\cite{Wal74} of an exchange economy.
Arrow and Debreu~\cite{AD54} proved, using the Kakutani
fixed-point theorem, that in a Walrasian market, if the utility
functions are continuous, strictly monotone, and quasi-concave,
then an equilibrium always exists. However, their proof gives
no indication of the market dynamics that might lead to an
equilibrium.

Markets are often classified according to the type of utility
functions used. Utility functions are usually assumed to be
concave and monotonically non-decreasing. (Economically,
the former assumption is the law of diminishing marginal
utility, and the latter is implied by assuming free disposal.)
A utility function of the form
$$
u(x) = \left(\sum_{j=1}^m a_{j}^\rho x_{j}^\rho\right)^{1/\rho},
$$
where $\rho\in (-\infty,0)\cup (0,1]$,
is called a utility with {\em constant elasticity of substitution}
or a CES utility. CES utilities with $\rho\in (0,1]$ are a special
case of utilities that satisfy the {\em weak gross substitutes}
(WGS) property---increasing the price of a good does not
decrease the demand for any other good. If $\rho = 1$, this
is a {\em linear} utility. In CES utilities with $\rho < 0$,
the goods are complementary. The limit of $u(x)$ as
$\rho\rightarrow -\infty$ is called a Leontief utility, and
the limit as $\rho\rightarrow 0$ is called a Cobb-Douglas
utility. A utility function of the form
$$
u(x) = \left(\sum_{j_1=1} ^k c_{j_1}^{\rho_1} \left(\sum_{j_2=1}^m a_{j_1,j_2}^{\rho_2} x_{j_1,j_2}^{\rho_2}\right)^{\rho_1/\rho_2}\right)^{1/\rho_1}
$$
is a (two-level) {\em nested} CES utility (see~\cite{Kel76}).
A {\em resource allocation} utility (see~\cite{KV02})
is a nested linear-Leontief utility ($\rho_1 = 1$ and $\rho_2 = -\infty$).
An {\em exponential} utility (a.k.a. constant absolute risk
aversion utility) is a simple example of a utility that does
not exhibit constant elasticity of substitution. It has the
form
$$
u(x) = \sum_{j=1}^m a_j \left(1 - e^{-\theta x_j}\right).
$$

\paragraph{T\^{a}tonnement.}
The idea of t\^{a}tonnement is due to Walras~\cite{Wal74}.
Arrow et al.~\cite{ABH59} showed that a time-continuous
version of t\^{a}tonnement converges to equilibrium for WGS
utilities. Scarf~\cite{Sca60} gave examples of Arrow-Debreu
markets with non-WGS utilities (in particular Leontief utilities)
where specific implementations of t\^{a}tonnement cycle and
never converge to an equilibrium.

The more recent computer science literature considers
discrete versions of t\^{a}tonnement in Fisher markets.
(Discrete t\^{a}tonnement was also discussed in the
economics literature, see~\cite{Uza60}.) Codenotti et
al.~\cite{CMV05} show that a t\^{a}tonnement-like process
involving some price coordination converges in
polynomial time for WGS utilities. Another t\^{a}tonnement-like
process that involves coordination is analyzed in Fleischer et
al.~\cite{FGKKS08}. They show, for many interesting markets,
a weak form of convergence of the process. They consider
the average of the prices and allocations over the steps of the
process and show that they converge to a weak notion of
approximate equilibrium, which satisfies only the sum of budget
constraints of the agents. On the other hand, their results
apply to a wide range of markets. In particular, they apply to
all the markets where we demonstrate explicit bounds on
convergence to our much stronger notion of approximate
equilibrium.

The first result to demonstrate the convergence to equilibrium
of a true discrete t\^{a}tonnement process is due to Cole and
Fleischer~\cite{CF08}, who showed that the prices converge
to equilibrium prices for non-linear CES utilities that satisfy
WGS (i.e., $\rho\in (0,1)$) and for Cobb-Douglas utilities.
They analyze the same price-adjustment process given in
Equation~\eqref{eq: linear tatonnement}.
Their analysis relies on a strong property that they prove: in
the cases they analyze the price of each good moves towards
the equilibrium value at each iteration. This is generally false
for non-WGS utilities. Nevertheless, Cheung et al.~\cite{CCR12}
modified the analysis to apply to some non-WGS utilities,
including complementary CES utilities with $\rho\in (-1,0)$,
two-level nested
CES utilities with $\rho_1,\rho_2 > -1$, and even multi-level
nested CES utilities with some restrictions on the elasticity
(which do not breach the lower bound of $-1$). It should be
noted that the analysis of~\cite{CF08,CCR12} applies also
to a model of asynchronous price adjustments. Asynchrony
is clearly a desirable property of a plausible dynamic for a
market economy (though, as mentioned above, t\^{a}tonnement
is perhaps not the right process to consider in that case) or
a true peer-to-peer system. Our results, while applicable to
a wider range of utility functions, assume synchronous price
adjustments.

Recently, Cheung et al.~\cite{CCD13} show that a discrete
t\^{a}tonnement process that uses an exponential function
update rule (Equation~\eqref{eq: linear tatonnement} can
be thought of as a linearization of their rule) converges to
the optimal value of the Eisenberg-Gale convex program,
for complementary CES utilities and for Leontief utilities, i.e.
for all $\rho\in [-\infty,0)$. We note that our
Lemma~\ref{lm: price convergence} shows that in the cases
that they analyze, the convergence of the convex program
objective function also implies convergence to equilibrium
(in the same sense that we use in our results).
The analysis of~\cite{CCD13} relates their process
to generalized gradient descent using a judicious choice of a
Bregman divergence. In comparison, our main result uses the
multiplicative weights
update paradigm (see~\cite{AHK12}). It is well-known that
multiplicative weights update analysis and gradient descent
are dual arguments. We note that
there is no Bregman divergence that yields the linear descent
step of Equation~\eqref{eq: linear tatonnement} (which is not
to say that the analysis of~\cite{CCD13} cannot be adapted
to apply to the linear update that we use here). We further
note that like our analysis, the results in~\cite{CCD13} only
apply to synchronous price adjustments.

\paragraph{Discussion and comparison.}
In Section~\ref{sec: preliminary} we give a proof of convergence
via gradient descent that applies in particular to the CES utilities 
analyzed in~\cite{CF08,CCR12,CCD13}, is somewhat
simpler than the proof in~\cite{CCD13}, and gives similar bounds
on the rate of convergence. Our analysis of nested 
CES-Leontief utilities in Section~\ref{sec: nested} is new. 
Our results for nested CES-Leontief
utilities give an upper bound on the convergence rate which is
proportional to $1/\delta^3$. For Leontief utilities, the analysis 
in~\cite{CCD13} combined 
with our Lemma~\ref{lm: price convergence} gives a better bound 
proportional to $1/\delta^2$.

The pervasive application of the multiplicative weights update method
uses $\eps$ that depends on the desired accuracy of the outcome,
and proves the convergence of a solution that averages over the
iterations. Nevertheless, we show convergence of the sequence
of prices (and allocations), and not just convergence of average prices
(clearly implied by the former), and we show this for a {\em fixed}
$\eps$ in Equation~\eqref{eq: linear tatonnement} that is independent
of the desired closeness to equilibrium.

The techniques in~\cite{FGKKS08} are somewhat similar to ours.
However, they only prove the
convergence of the average, and they use a much weaker notion
of approximate equilibrium. On the other hand, that allows them to
apply their results directly to utility functions for which our stronger
notion of convergence is unlikely to hold, such as linear or resource
allocation utilities. By tweaking the t\^atonnement process,
we can handle these utilities as well.
This is briefly discussed and analyzed in Section~\ref{sec: examples}.
In this context, we note that whether one is interested in
convergence of the sequence or in convergence of the average
depends on the phenomenon that is being modeled. One can think
of t\^{a}tonnement as a negotiation process---the price setters
propose prices and collect purchase offers to study the
market, while no exchange actually occurs during the process.
In this case, convergence of the sequence seems more suitable.
An alternative view of t\^{a}tonnement is that of an ``exploration"
process---the price setters actually sell a fraction of the supply
at each step of the process at the prices of that step. In this case,
one would focus on convergence of the average. (A more rigorous
treatment of this interpretation is given by the {\em ongoing markets}
model of~\cite{CF08,CCR12}.) Our proofs imply that taking the
average over the steps of the prices and the allocations yields a
somewhat relaxed notion of an approximate equilibrium.

Finally, we discuss the robustness of our results. For simplicity,
we assume that throughout the process the prices satisfy the
condition that their sum equals the total budget available to the
agents. This condition is easy to satisfy initially, without explicit
coordination of the prices. For example,
set each price to be arbitrary (but strictly positive), and at the
first round reset the price of each good to be the total money
spent on the good, given the initial arbitrary prices. Once the
assumption is
satisfied, it is maintained if the process is carried out precisely.
However, we note that even if the initial prices deviate from
satisfying this assumption, and the price adjustment or even
the responses of the agents are perturbed, the convergence theorems
would still hold under trivial modifications. (However, note that
the initial prices must be strictly positive---t\^{a}tonnement cannot
recover from zero prices). In particular, the proof of
Proposition~\ref{pr: tatonnement p} also implies that the prices
converge very quickly towards satisfying the above assumption.
Once the sum of prices is close to the sum of budgets, the rest
of the proof can be modified to the slight deviation from equality
and also to slight deviations from optimality of the responses
of the agents.

\section{Fisher Markets}\label{sec: Fisher}

Here we present the Fisher market model and some
basic definitions. In a Fisher market
there are $m$ perfectly divisible goods, each with
quantity scaled to $1$, without loss of generality.
There are $n$ agents. Each agent $i$ is endowed
with a budget of $b_i$, and aims at maximizing a
concave utility function $u_i:\RR_+^m\rightarrow\RR_+$.
We'll assume monotonicity, so each agent spends
all its budget. Given monotonicity, we may assume
without loss of generality that for every agent $i$,
$u_i(\vec{0}) = 0$. Also without loss of generality,
we will
assume throughout the paper that the budgets are
scaled so that $\sum_{i=1}^n b_i = 1$.

A market equilibrium is a pair $(p,x)$ where
$p:\{1,2,\dots,m\}\rightarrow\RR_+$
is an assignment of non-negative prices to the goods,
and $x:\{1,2,\dots,n\}\times\{1,2,\dots,m\}\rightarrow\RR_+$
is an allocation of goods to agents, satisfying the following
conditions: ($i$) The total spend $\sum_{j=1}^m p_j x_{ij}$
of agent $i$ is at most $b_i$. ($ii$) The basket of goods $x_i$
that agent $i$ gets maximizes the utility $u_i$ for any basket
whose cost is at most $b_i$. ($iii$) The total demand $\sum_{i=1}^n x_{ij}$
for good $j$ is at most $1$. ($iv$) If the total demand for good
$j$ is less than $1$, then $p_j = 0$.

Let $x_i(p)$ denote the optimal basket of goods maximizing
the utility $u_i$ of agent $i$, under the budget constraint $b_i$
and the market prices $p$. Notice that $x_i(p)$ is given
by a solution to the following convex program:
$$
x_i(p) = \arg\max_{x_i\in\RR_+^m} \left\{u_i(x_i):\
\sum_{j=1}^m p_j x_{ij} \le b_i\right\}.
$$

We assume that computing an optimal basket is a
tractable problem. Further denote by
$$
z_j(p)\corresponds \sum_{i=1}^n x_{ij}(p) - 1
$$
the excess demand for good $j$ under the prices $p$.
Notice that an equilibrium price vector $p^*$ satisfies $p^*\in B$, where
$$
B = \left\{p\in\RR_+^m:\ \sum_{j=1}^m p_j = \sum_{i=1}^n b_i\right\}.
$$
We now define the notion of approximate equilibrium. We give two
alternative definitions. In the first definition, each agent buys an
optimal basket of goods subject to the prices, but the demand for
each good may exceed the supply by a little. In the second definition,
each agent buys a near-optimal basket of goods, and the supply
constraints are satisfied. Definion~\ref{def: approx-equilib 1}
implies Definition~\ref{def: approx-equilib 2} (possibly with
a change of the approximation guarantee $\delta$) in many
interesting markets. This is the case when a small change in
allocation results in a small change in utility. Also,
Definition~\ref{def: approx-equilib 2}
implies Definition~\ref{def: approx-equilib 1} in many interesting
markets. Specifically, if
given specific prices, the optimal allocation for each player is
unique and the utility functions are strongly concave, then this
is the case. Definition~\ref{def: approx-equilib 2} is more natural
when the allocations
are very sensitive to small changes in the prices (one such
example is the case of linear utilities).

The definition that we use in most of the paper is:
\begin{definition}\label{def: approx-equilib 1}
A price-demand pair $(p,x)$ is a {\em $\delta$-approximate equilibrium} iff
\begin{itemize}
\item[P1.] For every agent $i=1,2,\dots,n$, the demand of $i$ optimizes
         the utility of $i$ given the prices $p$: $x_i = x_i(p)$.
\item[P2.] For every good $j=1,2,\dots,m$, the demand for $j$
         does not exceed the supply by much: $z_j(p)\le\delta$.
\item[P3.] The goods that are not purchased almost entirely do not
         cost much: For every $j=1,2,\dots,m$, if $z_j(p) < -\delta$
         then $p_j\le\delta$.
\end{itemize}
\end{definition}

The alternative definition that is often equivalent to the first one is:
\begin{definition}\label{def: approx-equilib 2}
A price-demand pair $(p,x)$ is a {\em $\delta$-approximate equilibrium} iff
\begin{itemize}
\item[P1.] For every agent $i=1,2,\dots,n$, the utility of $i$ is near-optimal
         given the prices $p$: $u_i(x_i)\ge (1 - \delta)\cdot u_i(x_i(p))$.
\item[P2.] For every good $j=1,2,\dots,m$, the demand for $j$
         does not exceed the supply: $z_j(p)\le 0$.
\item[P3.] The goods that are not purchased almost entirely do not
         cost much: For every $j=1,2,\dots,m$, if $z_j(p) < -\delta$
         then $p_j\le\delta$.
\end{itemize}
\end{definition}

For clarity, we use Dirac's notation to do linear algebra in
a vector space over $\RR$ endowed with the standard
Euclidean norm. In particular $\langle u \mid v\rangle$
denotes the inner product between two vectors $u$ and
$v$, and $\langle u \mid A \mid v\rangle$ denotes the
bilinear map of a pair of vectors $u,v$ using a linear
operator $A$. We denote by $\RR_+$ the set of non-negative
real numbers, and by $\RR_{++}$ the set of strictly
positive real numbers.

\section{The T\^{a}tonnement Process}\label{sec: tatonnement}

In this section, we discuss arbitrary Eisenberg-Gale
markets~\cite{JV10}. In an Eisenberg-Gale market, an equilibrium
allocation of goods to agents can be computed using a specific
convex program whose objective function is strictly concave in the
utilities of the agents. The Lagrange variables corresponding to
the supply constraints give the equilibrium prices. Thus, an
equilibrium price vector can be computed by solving the dual
program.

Let $x_i$ denote a feasible allocation for agent $i$,
so $x_{ij}$ is the quantity of good $j$ that agent
$i$ gets in the allocation $x_i$. An equilibrium
allocation $x^*$ is given by the solution to the
following Eisenberg-Gale convex program.
$$
\begin{array}{lll}
\max & \sum_{i=1}^n b_i \ln u_i(x_i) & \\
\hbox{s.t.} & \sum_{i=1}^n x_{ij} \le 1 & \forall j=1,2,\dots,m \\
 & x_{ij} \ge 0 & \forall i=1,2,\dots,n;\ \forall j=1,2,\dots,m.
\end{array}
$$
An equilibrium price vector $p^*$ is given by the solution
to the dual program:
$$
\begin{array}{lll}
\min & \sum_{j=1}^m p_j + \sum_{i=1}^n g_i^*(\mu_i) & \\
\hbox{s.t.} & p_j\ge -\mu_{ij} & i=1,2,\dots,n;\ \forall j=1,2,\dots,m \\
 & p_j \ge 0 & \forall j=1,2,\dots,m,
\end{array}
$$
where $g_i^*$ is the convex conjugate of the convex function
$g_i(x_i) = -b_i \ln u_i(x_i)$. More specifically, $g_i^*$ satisfies
$g_i^*(\mu_i) = \sup_{x\in\RR_+^m}\left\{\langle\mu_i \mid x\rangle - g_i(x)\right\}$.
Notice that $g_i^*$ is monotonically non-decreasing in each
coordinate of $\mu_i$. Therefore, we can replace $\mu_{ij}$
by $-p_j$ for every $i,j$.
For future use, we denote the primal objective function by
$\psi(x) = \sum_{i=1}^n b_i \ln u_i(x_i)$, and the dual objective
function by
$\phi(p) = \sum_{j=1}^m p_j + \sum_{i=1}^n g_i^*(-p)$.

The following lemma gives an explicit expression for the gradient
of $\phi$. It appears in~\cite[Lemma 3.3]{CCD13}.
In the appendix, we give the (short) proof for completeness.
\begin{lemma}[\cite{CCD13}]\label{lm: gradient}
In every Eisenberg-Gale market,
$$
\nabla\phi(p) = -z(p).
$$
\end{lemma}

We also need the following property of the primal objective function.
\begin{claim}\label{cl: primal approx}
For every $\alpha > 0$,
$\psi\left(\frac{1}{1+\alpha}\cdot x\right)\ge \psi(x) - \alpha$.
\end{claim}

\begin{proof}
For every agent $i$,
$$
u_i\left(\frac{1}{1+\alpha}\cdot x\right) =
u_i\left(\frac{1}{1+\alpha}\cdot x + \frac{\alpha}{1+\alpha}\cdot \vec{0}\right)\ge
\frac{1}{1+\alpha}\cdot u_i(x) + \frac{\alpha}{1+\alpha}\cdot u_i(\vec{0}) =
\frac{1}{1+\alpha}\cdot u_i(x),
$$
where the inequality follows from the concavity of $u_i$. Thus,
$$
\psi\left(\frac{1}{1+\alpha}\cdot x\right) =
\sum_{i=1}^n b_i\ln\left(u_i\left(\frac{1}{1+\alpha}\cdot x\right)\right)\ge
\sum_{i=1}^n b_i\ln\left(\frac{1}{1+\alpha}\cdot u_i(x)\right)\ge
\psi(x) - \alpha,
$$
where the second inequality follows from $\ln(1+\alpha)\le \alpha$
and the scaling $\sum_{i=1}^n b_i = 1$.
\end{proof}

The t\^{a}tonnement process proceeds as follows. Start
with an arbitrary assignment of strictly positive prices
$p^0\in\RR_{++}^m\cap B$. Given the time $t$ prices $p^t$,
each agent independently responds with its time $t+1$
demand $x_i^{t+1}\corresponds x_i(p^t)$. Given the time
$t+1$ demands $x^{t+1}$, the prices of goods are updated
according to the excess demand --- prices increase
for goods whose demand exceeds supply
and prices decrease for goods whose supply exceeds demand.
More specifically,
$$
p^{t+1}_j = p^t_j\cdot\left(1 + \veps z_j(p^t)\right),
$$
where $\veps\in (0,\frac 1 2)$ is a constant to be specified later.

We now give an alternative characterization of the optimal
response of the agents to a price vector. 
\begin{lemma}\label{lm: mwum response}
Let $p\in B$ be a price vector. Then,
$$
x(p) =
\arg\max_{x\in\RR_+^{n\times m}}\left\{\sum_{i=1}^n b_i\ln u_i(x_i):\
\frac{1}{\|p\|_1}\cdot\sum_{j=1}^m p_j\sum_{i=1}^n x_{ij}\le 1\right\}.
$$
\end{lemma}

\begin{proof}
The proof uses essentially the same argument as in
the proof of Lemma~\ref{lm: gradient}.

Given a price vector $p\in B$. Consider the following program,
$$\max \left\{\sum_{i=1}^n b_i\ln u_i(x_i):\
\frac{1}{\|p\|_1}\cdot\sum_{j=1}^m p_j\sum_{i=1}^n x_{ij}\le 1\right\}.
$$
By Lagrange duality we get that the maximum is achieved at
$$
\lambda_{\max},x_{\max}= \arg\max_{x\in\RR_+^{n\times m},\lambda\ge 0}\left\{\sum_{i=1}^n b_i\ln u_i(x_i) -\lambda \left(1-\sum_{j=1}^m\frac{p_j}{\|p\|_1}\sum_{i=1}^n x_{ij}\right)\right\}.
$$
We will show that $\lambda_{\max} =\|p\|_1$ and $x_{\max}=x(p)$ by
showing that they satisfy the KKT conditions. The first condition
is that if $\lambda_{\max} >0$  then
$1=\sum_{j=1}^m\frac{p_j}{\|p\|_1}\sum_{i=1}^n (x_{\max})_{ij}$.
By the definition of $x(p)$, $\sum_{j=1}^m p_j x(p)_{ij} = b_i$
(as $u_i(x)$ is monotonically non-decreasing). Thus, $x(p)$
satisfies $\sum_{i=1}^n\sum_{j=1}^m p_j x(p)_{ij} = \sum_{i=1}^n
b_i =\|p\|_1$. For the second KKT condition notice the following.
In a market in which every good $j$ has quantity $\sum_{i=1}^n
x(p)_{ij}$, the price-demand pair $(p,x(p))$ is an equilibrium
(each agent optimizes its demand and the market clears). Thus,
$x(p)$ is a primal optimal solution to the Eisenberg-Gale convex
program for this market. Therefore,
\begin{eqnarray*}
x(p)&=&\arg\max_{x\in\RR_+^{n\times m}}\left\{\sum_{i=1}^n b_i\ln u_i(x_i):
  \sum_{i=1}^n x_{ij} =\sum_{i=1}^n x(p)_{ij} \forall j\right\}\\
&=&\arg\max_{x\in\RR_+^{n\times m}}\left\{\sum_{i=1}^n b_i\ln u_i(x_i)-
\sum_{j=1}^m p_j \sum_{i=1}^n (x(p)_{ij}-x_{ij})\right\}\\
&=&\arg\max_{x\in\RR_+^{n\times m}}\left\{\sum_{i=1}^n b_i\ln u_i(x_i)-
\sum_{j=1}^m p_j \sum_{i=1}^n (1-x_{ij})\right\}.
\end{eqnarray*}
The first equation follows from the definition of the
Eisenberg-Gale market with quantities $\sum_{i=1}^n x(p)_{ij}$.
The second equation follows from Lagrange duality and the fact
that the prices $p$ are a dual optimal solution. The third
equation follows as $\sum_{j=1}^m p_j \sum_{i=1}^n x(p)_{ij}$ is a
constant independent of $x$.

Going back to the second KKT condition, notice that putting
$\lambda_{\max} = \|p\|_1$ gives exactly the last equation.
Therefore, $x_{\max}=x(p)$, $\lambda_{\max} = \|p\|_1$ satisfy the
second KKT condition as well.
\end{proof}

\begin{proposition}\label{pr: tatonnement p}
If $p^0\in B$ then $p^t\in B$ for all $t$.
\end{proposition}

\begin{proof}
The proof is by induction on $t$. For $t=0$, the claim
is satisfied by assumption.
For the inductive step, notice that
\begin{eqnarray*}
\sum_{j=1}^m p^{t+1}_j  & = & (1 - \eps) \sum_{j=1}^m p^t_j +
    \eps \sum_{j=1}^m p^t_j\sum_{i=1}^n x_{ij}^{t+1} \\
    &=& (1 - \eps) \sum_{j=1}^m p^t_j +
    \eps \sum_{i=1}^n \sum_{j=1}^m p^t_j x_{ij}(p^t)
=\sum_{j=1}^n b_i
\end{eqnarray*}
As we assume that $u_i(x_i)$ is monotonically non-decreasing,
by the definition of $x_i(p)$ for every $i$,
$\sum_{j=1}^m p^t_j x_{ij}(p^t) = b_i$ (every agent spends all its
budget at every round). Thus the last equality follows
(together with the induction hypothesis on $p^t$).
\end{proof}

\section{Preliminary Results}\label{sec: preliminary}

The following theorem generalizes and simplifies slightly 
the results of~\cite[Theorem 6.4]{CCD13} on the convergence 
of t\^{a}tonnement
for CES utility functions (albeit for the linearized version
of the process that they consider). 
\begin{theorem}\label{thm: strongly convex}
Suppose that there is a choice of $\veps > 0$ and a closed
convex set $P\subset\RR_+^m$
which includes an equilibrium price vector $p^*$ and the
sequence $p^0, p^1, p^2, \dots$ with the following properties:
($i$) There are bounds $p_{\min}\le p_{\max}$ such that
for all $p\in P$ and for all goods $j$, $p_j\in [p_{\min},p_{\max}]$.
($ii$) In $P$, the dual objective function $\phi$ is twice continuously
differentiable and strongly convex, and its gradient $\nabla\phi$ is
$L$-Lipschitz.
($iii$) $\veps\le \frac{p_{\min}}{p_{\max}^2 L}$.
Then, there exists a constant $\delta > 0$ such
that for every $T\ge 0$,
$$
\phi(p^T) - \phi(p^*)\le (1 - \delta)^{T}\cdot
\left(\phi(p^0) - \phi(p^*)\right).
$$
\end{theorem}

\begin{proof}
The proof closely mimics the proof of convergence of gradient
descent in the case of a strongly convex objective function.

Condition ($ii$) implies that there exist bounds
$0 < \lambda_{\min}\le\lambda_{\max}$ such that for
every $p\in P$ the eigenvalues of the Hessian matrix
$\nabla^2 \phi(p)$ lie in the interval $[\lambda_{\min},\lambda_{\max}]$.

Set $\veps = \frac{p_{\min}}{p_{\max}^2 \lambda_{\max}}$.
Consider a time step $t$. Consider a point $p\in P$.
The second order Taylor expansion of $\phi(p)$ with
respect to $\phi(p^{t-1})$ gives
\begin{eqnarray*}
\phi(p) & = & \phi(p^{t-1}) + \langle \nabla\phi(p^{t-1}), p - p^{t-1}\rangle + 
        \frac 1 2 \cdot \langle p - p^{t-1} \mid \nabla\phi(q) \mid p - p^{t-1}\rangle,
\end{eqnarray*}
for some interpolation point $q$ between $p$ and $p^{t-1}$.
As $q\in P$ (because $P$ is convex), we have that the last
term satisfies
$$
\lambda_{\min}\cdot \|p - p^{t-1}\|_2^2\le
\langle p - p^{t-1} \mid \nabla\phi(q) \mid p - p^{t-1}\rangle\le
\lambda_{\max}\cdot \|p - p^{t-1}\|_2^2
$$

Thus, on the one hand,
\begin{eqnarray*}
\phi(p^t) & \le & \phi(p^{t-1}) + \langle \nabla\phi(p^{t-1}), p^t - p^{t-1}\rangle +
     \frac{\lambda_{\max}}{2}\cdot \|p^t - p^{t-1}\|_2^2 \\
& = & \phi(p^{t-1}) - \eps\cdot\sum_j p^{t-1}_j \left(\nabla\phi(p^{t-1})\right)_j^2 +
       \frac{\lambda_{\max}}{2}\cdot\eps^2\cdot\sum_j \left(p^{t-1}_j \nabla\phi(p^{t-1})_j\right)^2 \\
& \le & \phi(p^{t-1}) - \eps\cdot p_{\min}\cdot \|\nabla\phi(p^{t-1})\|_2^2 +
       \frac{\lambda_{\max}}{2}\cdot\eps^2\cdot p_{\max}^2\cdot \|\nabla\phi(p^{t-1})\|_2^2 \\
& \le & \phi(p^{t-1}) - \left(\frac{p_{\min}}{p_{\max}}\right)^2\cdot \frac{1}{2\lambda_{\max}}\cdot 
                \|\nabla\phi(p^{t-1})\|_2^2.
\end{eqnarray*}

On the other hand, for every $p\in P$,
\begin{eqnarray*}
\phi(p) & \ge & \phi(p^{t-1}) + \langle \nabla\phi(p^{t-1}), p - p^{t-1}\rangle + 
       \frac{\lambda_{\min}}{2}\cdot \|p - p^{t-1}\|_2^2.
\end{eqnarray*}
The right-hand side is minimized at 
$p = p^{t-1} - \frac{1}{\lambda_{\min}}\cdot \nabla\phi(p^{t-1})$, so we 
get that for every $p$,
\begin{eqnarray*}
\phi(p) & \ge & \phi(p^{t-1}) + \langle \nabla\phi(p^{t-1}), 
      -\frac{1}{\lambda_{\min}}\cdot \nabla\phi(p^{t-1})\rangle + 
\frac{\lambda_{\min}}{2}\cdot \left\|-\frac{1}{\lambda_{\min}}\cdot \nabla\phi(p^{t-1})\right\|_2^2 \\
& = & \phi(p^{t-1}) - \frac{1}{2\lambda_{\min}}\cdot \|\nabla\phi(p^{t-1})\|_2^2.
\end{eqnarray*}

By setting $p = p^*$ on the left-hand side, we get
\begin{eqnarray*}
\|\nabla\phi(p^{t-1})\|_2^2 & \ge & 2\lambda_{\min}\cdot\left(\phi(p^{t-1}) - \phi(p^*)\right).
\end{eqnarray*}

Subtracting $\phi(p^*)$ on both sides in the first bound and combining the two
bounds we get that
\begin{eqnarray*}
\phi(p^t) - \phi(p^*) & \le & \left(1 - \left(\frac{p_{\min}}{p_{\max}}\right)^2\cdot \frac{\lambda_{\min}}{\lambda_{\max}}\cdot\right)\cdot\left(\phi(p^{t-1}) - \phi(p^*)\right).
\end{eqnarray*}
This recurrence relation implies the theorem.
\end{proof}

The following lemma relates the proximity to the dual objective
optimum to the proximity to equilibrium. 
\begin{lemma}\label{lm: price convergence}
Under the same assumptions and notation of 
Theorem~\ref{thm: strongly convex},
for every $\delta > 0$, for every
$p\in P$, and for the dual optimal solution $p^*$,
if $\phi(p)\le \phi(p^*) + 
\min\left\{1,\frac{1}{\lambda_{\max}^2}\right\}\cdot
\frac 1 2\lambda_{\min}\delta^2$,
then $\|p - p^*\|_2\le\frac{\delta}{\lambda_{\max}}$ 
and $|z(p) - z(p^*)|_\infty\le \delta$.
\end{lemma}

\begin{proof}
Consider the second order Taylor expansion of $\phi(p)$
with respect to $\phi(p^*)$. For an interpolation point
$p' = \gamma p + (1-\gamma) p^*$,
$$
\phi(p) = \phi(p^*) + \langle p - p^* \mid \nabla\phi(p^*) \rangle +
  \frac 1 2 \langle p - p^* \mid \nabla^2\phi(p') \mid p - p^* \rangle.
$$
Notice that by the strong convexity assumption, $\lambda_{\min} > 0$.
Since $\nabla\phi(p^*) = \vec{0}$, we get that
$$
\min\{1,1/\lambda_{\max}^2\}\cdot\frac 1 2\lambda_{\min} \delta^2 \ge \phi(p) - \phi(p^*) =
\frac 1 2 \langle p - p^* \mid \nabla^2\phi(p') \mid p - p^* \rangle \ge
\frac 1 2\cdot\lambda_{\min}\cdot \|p - p^*\|_2^2.
$$
In particular, we get that $\|p - p^*\|_2^2\le\frac{ \delta^2}{\lambda_{\max}^2}$.
On the other hand, $|z_j(p) - z_j(p^*)| = |f_j(p) - f_j(p^*)|\le \|f(p) - f(p^*)\|_2$.
Consider $q(\gamma) = \gamma p + (1-\gamma) p^*$. By the fundamental
theorem of line integrals (a.k.a. the gradient theorem),
$f_j(p) - f_j(p^*) = \int_0^1 \langle \nabla(f_j(q(\gamma)))\mid p - p^*\rangle d\gamma$.
Therefore,
$$
\|f(p) - f(p^*)\|_2^2 = \sum_j \left(
    \int_0^1 \langle \nabla(f_j(q(\gamma)))\mid p - p^*\rangle d\gamma\right)^2
\le \sum_j \int_0^1\left(\langle \nabla(f_j(q(\gamma)))\mid p - p^*\rangle\right)^2 d\gamma.
$$
(The inequality is a Cauchy-Schwartz argument---consider the random variable
$\langle \nabla(f_j(q(\gamma)))\mid p - p^*\rangle$ on $\gamma\in [0,1]$ endowed
with the uniform probability measure.)
Therefore, there exists $\gamma\in [0,1]$ such that for $q = q(\gamma)$ we
have
$$
\|f(p) - f(p^*)\|_2^2\le \sum_j \left(\langle \nabla(f_j(q))\mid p - p^*\rangle\right)^2
= \|\nabla^2\phi(q) \mid p - p^*\rangle\|_2^2
\le \lambda_{\max}^2 \|p - p^*\|_2^2\le\delta^2.
$$
This completes the proof.
\end{proof}

\begin{corollary}\label{cor: convergence to equilibrium}
For every sufficiently small $\delta > 0$, 
if $\phi(p)\le \phi(p^*) + 
\min\left\{1,\frac{1}{\lambda_{\max}^2}\right\}\cdot
\frac 1 2\lambda_{\min}\delta^2$,
then $(p, x(p))$ is a $\delta$-approximate equilibrium
in the sense of Definition~\ref{def: approx-equilib 1}.
\end{corollary}

A market with (non-linear) CES utilities is defined by a constant
$\rho\in (-\infty,0)\cup (0,1)$,
and constants $c_{ij}\ge 0$, for $i=1,2,\dots,n$ and $j=1,2,\dots,m$.
The utility function of agent $i$ is
$u_i(x) = 
\left(\sum_{j=1}^m \left(c_{ij} x_{ij}\right)^{\rho}\right)^{1/\rho}$.\footnote{We note
that the definition and the following claim can be generalized easily to 
the case where for each player $i$ there is a different constant $\rho_i$ 
instead of one uniform constant $\rho$.}
\begin{corollary}[\cite{CCD13}]\label{cor: CES}
Consider a market with non-linear CES utilities.
There are constants $\kappa_1 = \kappa_1(b,c,\rho)$ and 
$\kappa_2 = \kappa_2(b,c,\rho)$ such that the following
holds. If $\veps\le \frac{1}{\kappa_1}$, then for every
$\delta > 0$ and for every
$T\ge \frac{-\kappa_2\ln p^0_{\min}}{\veps^2}\cdot \ln\delta$,
the price-demand pair $(p^T,x^{T+1})$ is
a $\delta$-approximate equilibrium in the sense
of Definition~\ref{def: approx-equilib 1}.
\end{corollary}

The limit of a CES utility
$u_i(x) = \left(\sum_{j=1}^m \left(c_{ij} x_{ij}\right)^{\rho}\right)^{1/\rho}$
as $\rho_i\rightarrow 0$ is called a Cobb-Douglas
utility. Explicitly, this puts
$u_i(x) = \prod_{j=1}^m x_{ij}^{c_{ij}}$.
Markets with Cobb-Douglas utilities converge to equilibrium
in one step if we set $\veps = 1$. For a fixed $\veps < 1$,
convergence to a $\delta$-approximate equilibrium in
$O(\log(1/\delta))$ steps is easy to establish using the
quantitative estimates of Banach's fixed-point theorem,
as the mapping $p^t\mapsto p^{t+1}$ is a contraction
(using any choice of norm). Alternatively, 
Theorem~\ref{thm: strongly convex} applied to Cobb-Douglas
utilities as well.
\begin{corollary}[\cite{CF08}]\label{cor: CobbDouglas}
Corollary~\ref{cor: CES} applies also to markets
with Cobb-Douglas utilities.
\end{corollary}

Some conventional utility functions, for example Leontief
utilities, lead to markets that do not satisfy the conditions
of Theorem~\ref{thm: strongly convex}. Leontief utilities
specifically are handled in~\cite{CCD13} using the (weaker)
bounds on the convergence of gradient descent for more
general convex functions. We offer here an alternative
approach using two ideas. One is a projection of the space
of prices that enables us to replace the bounds implied by
strong convexity by similar bounds that do not require
strong convexity. The other is a convergence analysis
of the multiplicative weights update method variety (which
is known to be equivalent to gradient descent---in fact it
is a dual argument). More specifically, we use the
following lemma, which is a variant of the multiplicative
weights update method.
Its proof appears in the appendix.
\begin{lemma}\label{lm: mwum}
Let $w = \max_{j,t} |z_j(p^t)|$, and let $\veps \le \frac{1}{2w}$.
Also let $v = \max_j \frac{1}{T}\cdot\sum_{t=0}^{T-1} (z_j(p^t))^2$.
Then,
$$
\frac{1}{T}\cdot\sum_{t=0}^{T-1} z_j(p^t)\le\veps v + \frac{\ln(1/p^0_j)}{\veps T}.
$$
\end{lemma}

Using Lemma~\ref{lm: mwum} directly to achieve convergence to
equilibrium within an error parameter $\delta$ requires a choice of
$\veps$ that depends on $\delta$. However, with some additional
assumptions, one can use a fixed $\veps$, independent of $\delta$.
This is what we discuss next. We will need the following definition.
\begin{definition}\label{def: uniform}
Let $\beta$ be a function mapping $\RR_+$ to itself.
An Eisenberg-Gale market is {\em $\beta$-uniform}
in a convex region $P$ of price vectors (which contains
an equilibrium vector $p^*$)
iff for every $\alpha > 0$ and for every $p\in P$ the
following holds.
If $\phi(p)\le \phi(p^*)+\beta(\alpha)$, then
$|z_j(p) - z_j(p^*)|\le \alpha$ for every $j=1,2,\dots,m$.
\end{definition}
\noindent Notice that Lemma~\ref{lm: price convergence} shows
that if $\phi$ is strongly convex, then the market is
$\beta$-uniform for $\beta(\alpha) = \Theta(\alpha^2)$.
Put $W = \max\{1,\max_{p\in P}\max_j z_j(p)\}$. 
\begin{lemma}\label{lm: better mwum}
Consider a $\beta$-uniform Eisenberg-Gale
market. Suppose that there exists a constant
$\kappa > 0$ such that for every $\alpha$,
$\beta(\alpha)\ge \kappa\cdot\alpha^2$. Further
suppose that for sufficiently small
$\veps$, the sequence $\phi(p^0),\phi(p^1),\dots$
is monotonically non-increasing. Then, for sufficiently
small $\veps$ the following holds.
For every $\delta > 0$, if
$T \ge\frac{W\ln(1/p^0_{\min})}{8\veps^2 \delta^3}$,
then $\frac{1}{T}\cdot\sum_{t=0}^{T-1} z_j(p^t)\le \beta(\delta)$.
\end{lemma}

\begin{proof}
We assume that $\veps\le\frac{1}{2W}\le\frac{1}{2w}$, and
also that $\veps\le \frac{\kappa}{64}$.
Trivially, for every $T$,
$$
\frac{1}{T}\cdot\sum_{t=0}^{T-1} (z_j(p^t))^2\le w^2.
$$
By Lemma~\ref{lm: mwum} we conclude that
for every $T\ge T_0 = \frac{\ln(1/p^0_{\min})}{\veps^2 w^2}$, 
$$
\frac{1}{T}\cdot\sum_{t=0}^{T-1} z_j(p^t)\le 2\veps w^2\le\beta(w/4),
$$
where the second inequality follows from $\veps\le \frac{\kappa}{64}$,
with room to spare. This means that 
$\frac{1}{1+\beta(w/4)}\cdot \frac{1}{T}\cdot\sum_{t=1}^{T} x^t$ is 
a feasible primal solution of the Eisenberg-Gale convex program. 
Thus,
\begin{eqnarray*}
\phi(p^*) & \ge & \psi\left(\frac{1}{1+\beta(w/4)}\cdot\frac{1}{T}\cdot
                     \sum_{t=1}^{T} x^t\right) \\
& \ge & \psi\left(\frac{1}{T}\cdot \sum_{t=1}^{T} x^t\right) - \beta(w/4) \\
& \ge & \frac{1}{T}\cdot\sum_{t=1}^{T} \psi(x^t) - \beta(w/4) \\
& = & \frac{1}{T}\cdot\sum_{t=0}^{T-1} \phi(p^t) - \beta(w/4) \\
& \ge & \phi(\frac{1}{T}\cdot\sum_{t=0}^{T-1} \phi(p^t)) - \beta(w/4) \\
& \ge & \phi(p^{T-1}) - \beta(w/4),
\end{eqnarray*}
where the first inequality uses weak duality, the second inequality
uses Claim~\ref{cl: primal approx}, the third inequality uses the
concavity of $\psi$, the equation follows by construction
($\psi(x^{t+1}) = \phi(p^t)$; see the proof of Lemma~\ref{lm: gradient}),
the fourth inequality uses the convexity of $\phi$, and the fifth
inequality uses the monotonicity of the sequence $\phi(p^0),\phi(p^1),\dots$.
We conclude that for every $T\ge T_0$,
$\phi(p^T) - \phi(p^*)\le \beta(w/4)$, and therefore
the $\beta$-uniformity of the market implies that
$\max_j |z_j(p^T) - z_j(p^*)|\le \frac{w}{4}$, and
in particular $\max_j z_j(p^T)\le\frac{w}{4}$.

Next consider
$T\ge T_0 + T_1$, where $T_1 = 7T_0$.
Consider $j\in\{1,2,\dots,m\}$. If $z_j(p^*) < -\frac{w}{4}$, then for every
$t > T_0$ , $z_j(p^t) < 0$, so
$$
\frac{1}{T}\cdot\sum_{t=0}^{T-1} z_j(p^t) < \frac{1}{T}\cdot\sum_{t=0}^{T_0}z_j(p^t)\le
\frac 1 8\cdot 2\veps w^2\le\beta(w/8).
$$
Here again we've used in the last inequality $\veps\le \frac{\kappa}{64}$.
Otherwise, if $z_j(p^*) \ge  -\frac{w}{4}$, then for every $t > T_0$,
$z_j(p^t)\ge z_j(p^*) - \frac{w}{4}\ge -\frac{w}{2}$. 
Also, for all $j$, $z_j(p^t)\le \frac{w}{4}$.
Therefore,
$$
\frac{1}{T}\cdot\sum_{t=0}^{T-1} (z_j(p^t))^2\le \frac 1 8 \cdot w^2  +
\frac 7 8\cdot (w/2)^2 < 2(w/2)^2.
$$
Moreover, $T\ge 8T_0 = 8\ln(1/p^0_{\min}) / \veps^2 w^2$,
so $\ln(1/p^0_{\min}) / \veps T\le \veps (w/2)^2$.
Therefore, by Lemma~\ref{lm: mwum},
$$
\frac{1}{T}\cdot\sum_{t=0}^{T-1} z_j(p^t)\le 3\veps (w/2)^2\le
\beta(w/8),
$$
where the last inequality uses $\veps\le \frac{\kappa}{64}$. 
Using the same argument that we used for $T\ge T_0$, we 
can conclude that for $T\ge T_0 + T_1$,
$\max_j |z_j(p^T) - z_j(p^*)|\le \frac{w}{8}$, so also
$\max_j z_j(p^T)\le \frac{w}{8}$.

More generally, suppose that we've verified that for every
$T\ge T_0 + T_1 + \cdots + T_i$,
$$
\frac{1}{T}\cdot\sum_{t=0}^{T-1} z_j(p^t)\le 3\veps (w/2^i)^2.
$$
As $3\veps (w/2^i)^2\le \beta(w/2^{i+2})$, we get that
$\max_j |z_j(p^T) - z_j(p^*)|\le w / 2^{i+2}$ and $\max_j z_j(p^T)\le w/2^{i+2}$.
This immediately implies that for every $j$, one of the following two cases holds:

\noindent Case 1: $z_j(p^T) < 0$ for all $T\ge T_0 + T_1 + \cdots + T_i$.

\noindent Case 2: $-w/2^{i+1}\le z_j(p^T)\le w/2^{i+2}$ for all 
$T\ge T_0 + T_1 + \cdots + T_i$.

Setting $T_{i+1} = 7(T_0+T_1+\cdots+T_i)$, we satisfy the
inductive hypothesis as follows. Consider
$T\ge T_0 + T_1 + \cdots + T_{i+1}$ and $j\in\{1,2,\dots,m\}$.
If case 1 holds, then
$$
\frac{1}{T}\cdot\sum_{t=0}^{T-1} z_j(p^t)\le
\frac{1}{T}\cdot\sum_{t=0}^{T_0+\cdots+T_i} z_j(p^t)\le
\frac 3 8\cdot\veps (w/2^i)^2 < 3\veps (w/2^{i+1})^2.
$$
If case 2 holds, then
$$
\frac{1}{T}\cdot\sum_{t=0}^{T-1} (z_j(p^t))^2\le
\frac{T_0 w^2 + T_1 (w/2)^2 + \cdots + T_i (w/2^i)^2 + T_{i+1} (w/2^{i+1})^2}
{T_0 + T_1 + \cdots + T_{i+1}}\le
$$
$$
\le\frac 7 8\cdot (w/2^{i+1})^2 + \frac 7 8\cdot\frac 1 8\cdot (w/2^i)^2 +
\frac 7 8\cdot\frac 1 8\cdot\frac 1 8 (w/2^{i-1})^2 + \cdots
< 2(w/2^{i+1})^2.
$$
Also,
$T\ge 8(T_0+\cdots+T_i) \ge \ln(1/p^0_{\min}) / \veps^2 (w/2^{i+1})^2$,
so $\ln(1/p^0_{\min}) / \veps T\le \veps (w/2^{i+1})^2$. Thus,
by Lemma~\ref{lm: mwum},
$$
\frac{1}{T}\cdot\sum_{t=0}^{T-1} z_j(p^t)\le 3\veps (w/2^{i+1})^2.
$$
This asserts the inductive hypothesis.

Finally, notice that if we choose
$i_{\max}\ge\log_2(w/\delta)-2$ we get that
for every $T\ge T_0 + T_1 + \cdots + T_{i_{\max}}$,
$$
\frac{1}{T}\cdot\sum_{t=0}^{T-1} z_j(p^t)\le\beta(\delta).
$$
Finally, notice that
$$
T_0 + \cdots + T_{i_{\max}}  <
8^{i_{\max}+1}\cdot\frac{\ln(1/p^0_{\min})}{\veps^2 w^2} \le
\frac{w\ln(1/p^0_{\min})}{8\veps^2 \delta^3}.
$$
This concludes the proof of the lemma.
\end{proof}

The following lemma gives rather general conditions for
the convergence of t\^{a}tonnement in Fisher markets.
\begin{lemma}\label{lm: main}
Consider a $\beta$-uniform Eisenberg-Gale
market. Fix $\delta > 0$.
Suppose that $\veps$ is sufficiently small and $T$ is sufficiently
large so that the following conditions hold.
\begin{itemize}
\item[C1.] The sequence $\phi(p^0),\phi(p^1),\dots,\phi(p^T)$ is 
                monotonically non-increasing.
\item[C2.] For every $T'=T-\left\lceil 3\ln(1/\delta)/\delta\eps \right\rceil,\dots,T$
                and for every $j=1,2,\dots,m$,
                \begin{equation}\label{eq: sum of z-s}
                \frac{1}{T'+1}\cdot\sum_{t=0}^{T'} z_j(p^t)\le \beta(\delta/3).
                \end{equation}
\end{itemize}
Then, the price-demand pair $(p^T,x^{T+1})$ is a $\delta$-approximate
equilibrium in the sense of Definition~\ref{def: approx-equilib 1}.
\end{lemma}

\begin{proof}
By construction, $x^{T+1} = x(p^T)$.

Fix $k\in\{0,1,2,\dots,\left\lceil 3\ln(1/\delta)/\delta\veps \right\rceil\}$.
Let $\E x^t = \frac{1}{T-k+1}\cdot\sum_{t=1}^{T-k+1} x^t$ and
let $\E p^t = \frac{1}{T-k+1}\cdot\sum_{t=0}^{T-k} p^t$.
By condition C2, $\frac{1}{1 + \beta(\delta/3)}\cdot\E x^t$
is a feasible primal solution. Therefore, we argue as in
the proof of Lemma~\ref{lm: better mwum} that
\begin{eqnarray*}
\phi(p^*) & \ge & \psi\left(\frac{1}{1+\beta(\delta/3)}\cdot\E x^t\right)
\ge \psi\left(\E x^t\right) - \beta(\delta/3) \ge
    \frac{1}{T-k+1}\cdot\sum_{t=1}^{T-k+1} \psi(x^t) - \beta(\delta/3) \\
& = &
\frac{1}{T-k+1}\cdot\sum_{t=0}^{T-k} \phi(p^t) - \beta(\delta/3) \ge
\phi(\E p^t) - \beta(\delta/3) \ge \phi(p^{T-k}) - \beta(\delta/3).
\end{eqnarray*}
The first inequality follows from weak duality,
the second inequality follows from Claim~\ref{cl: primal approx},
the third inequality follows from the concavity of $\psi$,
the equality follows by construction,
the fourth inequality follows from the convexity of $\phi$,
and the last inequality follows from condition C1 in the theorem.
Rearranging the terms, we get that $\phi(p^{T-k})\le\phi(p^*) + \beta(\delta/3)$.
Thus, because the market is $\beta$-uniform, we have
that $|z_j(p^{T-k}) - z_j(p^*)|\le \frac{\delta}{3}$ for every $j=1,2,\dots,m$.

Consider $j\in\{1,2,\dots,m\}$. Notice that $z_j(p^*)\le 0$.
If $z_j(p^*)\ge -\frac{2\delta}{3}$ then since for $k=0$ the
above argument gives
$|z_j(p^T) - z_j(p^*)|\le\frac{\delta}{3}$,
we get that $|z_j(p^T)|\le\delta$.
Otherwise, we have that
$z_j(p^{T-k}) < -\frac{\delta}{3}$, for
$k=0,1,2,\dots,\left\lceil 3\ln(1/\delta)/\delta\veps \right\rceil$.
But, $p_j^{T-\left\lceil 3\ln(1/\delta)/\delta\veps \right\rceil}\le 1$, and for every
$k\in\{0,1,2,\dots,\left\lceil 3\ln(1/\delta)/\delta\veps \right\rceil-1\}$,
$p_j^{T-k}\le (1 - \delta\veps/3) p_j^{T-k-1}$. Thus,
$p_j^T\le (1 - \delta\veps/3)^{\left\lceil 3\ln(1/\delta)/\delta\veps \right\rceil}\le\delta$.
\end{proof}

When we apply Lemma~\ref{lm: main} (using Lemma~\ref{lm: better mwum}
to establish Condition C2) to Leontief utilities, our approach gives somewhat
weaker bounds on the convergence rate than those in~\cite{CCD13}.
However, we can use this approach to handle the more general
case of nested CES-Leontief utilities, which is our main result.

\section{Nested CES-Leontief Utilities}\label{sec: nested}

In a market with nested CES-Leontief utilities
every agent
$i=1,2,\dots,n$ needs a set of
``objects" ${\cal J}_i$. Each object $J\in{\cal J}_i$
consists of a utility coefficient $c_i^J > 0$, a set of goods (which
we also denote by $J$), and utility coefficients $a_{ij}^J > 0$
for all $j\in J$. (To simplfy some of the expressions below we set
$a_{ij}^J = 0$ for all $j\not\in J$.) The utility functions are formally given by
$u_i(x) = \left(\sum_{J\in{\cal J}_i} \left(c_i^J \min_{j\in J}
\left\{\frac{x_{ij}^J}{a_{ij}^J}\right\}\right)^\rho\right)^{1/\rho}$,
for some $\rho\in (-\infty,0)\cup (0,1)$.
W.l.o.g. we scale $a$ and $c$ so that
$\|a\|_1 = 1$ and for every $i$, $\sum_{J\in{\cal J}_i}(c_i^J)^{\rho/(1-\rho)}=1$.
(Notice that the behavior of agent $i$ depends only on the relative values
of the coordinates of $a_i$ and $c_i$.) Also, we may assume that for
every good $j$ there is at least one buyer $i$ and at least one object
$J\in{\cal J}_i$ for which $j\in J$, otherwise the good has no demand
and is can be discarded.
The Eisenberg-Gale convex program is given explicitly as:
$$
\max \left\{\sum_i b_i \ln \left(\sum_{J\in{\cal J}_i}
\left(c_i^J u_i^J\right)^\rho\right)^{1/\rho} :
\forall j=1,\dots, m,\ \sum_{i=1}^{n}\sum_{J\in{\cal J}_i} a_{ij}^J u_i^J \le 1
\wedge \forall i=1,\dots,n,\forall J\in{\cal J}_i,\ u_i^J\ge 0\right\}.
$$
The dual objective is
$$
\phi(p) = \sum_{j=1}^m p_j - \sum_{i=1}^n b_i +
\sum_{i=1}^n b_i\ln\left(\sum_{J\in{\cal J}_i}
\left(\frac{b_i c^J_{i}}{\sum_{j=1}^m a^J_{ij} p_j}
\right)^{\rho/(1-\rho)}\right)^{(1-\rho)/\rho}.
$$
For a vector $v\in\RR^m$, $i\in\{1,2,\dots,n\}$, and
$J\in{\cal J}_i$, denote $\tilde{v}_i^J = \sum_{j=1}^m a_{ij}^J v_j$.
Also denote $a_{\min} = \min_{i,J\in {\cal J}_i,j\in J} a_{ij}^J$,
$A = \max_j \sum_{i,J\in {\cal J}_i} a_{ij}^J$,
$b_{\min} = \min_i b_i$, and
$c_{\min}=\min_{i,J\in {\cal J}_i} (c_i^J)^{\frac{\rho}{1-\rho}}$.
\begin{claim}\label{cl: tilde norm}
For every $v\in\RR^m$, $a_{\min}^2\cdot \|v\|_2^2\le \|\tilde{v}\|_2^2\le A\cdot\|v\|_2^2$.
\end{claim}

\begin{proof}
By definition, $\|\tilde{v}\|_2^2 = \sum_i\sum_{J\in{\cal J}_i}
\left(\sum_{j\in J} a_{ij}^J v_j\right)^2$. As every good $j$ has
an object $J$ for which $j\in J$, then trivially
$\sum_i\sum_{J\in{\cal J}_i}
\left(\sum_{j\in J} a_{ij}^J v_j\right)^2\ge a_{\min}^2 \|v\|_2^2$.
On the other hand, using Cauchy-Schwartz,
$\sum_i\sum_{J\in{\cal J}_i}
\left(\sum_{j\in J} a_{ij}^J v_j\right)^2\le\sum_i\sum_{J\in{\cal J}_i}
\sum_{j\in J} a_{ij}^J v_j^2 =
\sum_j v_j^2\cdot\sum_i\sum_{J\in {\cal J}_i} a_{ij}^J\le
A\cdot\|v\|_2^2$.
\end{proof}

We also need the following claim.
\begin{claim}\label{cl: price bounds}
If for every agent $i$ and for every $J\in {\cal J}_i$,
$$
\tilde{p}_i^J\ge \left\{\begin{array}{ll}
                                 2^{(\rho-2)/(1-\rho)}\cdot b_i\cdot c_{\min}^2 & \hbox{if }\rho > 0, \\
                                 2^{\rho-2} b_i^{1-\rho} & \hbox{if }\rho < 0,
                                 \end{array}
                         \right.
$$
where $p = p^0$, then the same bounds hold for $p = p^t$, for all $t$.
\end{claim}

\begin{proof}
Notice that for every agent $i$, the demand that $i$ has for the object
$J\in {\cal J}_i$ given the prices $p$ is
$\frac{b_i}{\tilde{p}_i^J}\cdot
\frac{(c_i^J / \tilde{p}_i^J)^{\frac{\rho}{1-\rho}}}
{\sum_{J'\in {\cal J}_i} (c_i^{J'} / \tilde{p}_i^{J'})^{\frac{\rho}{1-\rho}}}$.
Consider $J_0\in {\cal J}_i$ that minimizes $\tilde{p}_i^{J_0}$, and let
$J_1\in {\cal J}$ be any object for which $\tilde{p}_i^{J_1}\le 2\tilde{p}_i^{J_0}$.

Let's first consider $\rho > 0$.
Using the scaling of $c_i$, we have that
$\sum_{J\in {\cal J}_i} (c_i^J / \tilde{p}_i^{J})^{\frac{\rho}{1-\rho}}\le
|{\cal J}_i| / (\tilde{p}_i^{J_0})^{\frac{\rho}{1-\rho}}$ and
$c_{\min}\le\frac{1}{|{\cal J}_i|}$.
Therefore, the demand that $i$ has for $J_1$
is at least $\frac{b_i\cdot c_{\min}^2}{2^{1/(1-\rho)}\tilde{p}_i^{J_0}}$.
This is a lower bound on $x(p)_{ij}$ for every $j\in J_1$.
Notice that if for all $j\in J_1$, $x(p)_{ij} > 1$, then
$\tilde{p}_i^{J_1}$ must increase in the next time step.
Thus, we conclude that if
$\tilde{p}_i^{J_0} < 2^{1/(1-\rho)}\cdot b_i\cdot c_{\min}^2$
then $\tilde{p}_i^{J_1}$ increases in the next time step.

A similar analysis applies to $\rho < 0$. In this case we
bound the demand that agent $i$ has for $J_1$ by using the
fact that for every $J\in {\cal J}_i$, $\tilde{p}_i^{J}\le 1$.
Therefore,
$\sum_{J\in {\cal J}_i} (\tilde{p}_i^{J} / c_i^J)^{\frac{-\rho}{1-\rho}}\le
|{\cal J}_i|\cdot c_{\min}\le 1$. Using the fact that
$(c_i^{J_1})^{\rho/(1-\rho)}\le 1$ and
$\tilde{p}_i^{J_0}\le \tilde{p}_i^{J_1}\le 2\tilde{p}_i^{J_0}$,
we get that the demand is at least $\frac{b_i}{2(\tilde{p}_i^{J_0})^{1/(1-\rho)}}$.
So if $\tilde{p}_i^{J_0} < 2^{\rho - 1}\cdot b_i^{1-\rho}$ then
$\tilde{p}_i^{J_1}$ increases in the next time step.

Now the rest of the proof follows by induction on the
number of time steps, assuming that the inequality
holds initially. Notice that in one iteration the prices
never drop by more than a factor of $2$ (because
$\veps\le\frac 1 2$). So let $q$ denote the new prices,
and let $J_1$ denote the object minimizing $\tilde{q}_i^{J}$
(over $J\in {\cal J}_i$). If $\tilde{p}_i^{J_1} > 2\tilde{p}_i^{J_0}$,
then the induction hypothesis holds trivially. Otherwise,
we showed that there exists $\gamma$ such that
if $\tilde{p}_i^{J_0} < \gamma$ then $\tilde{q}_i^{J_1} > \tilde{p}_i^{J_1}$.
If indeed $\tilde{p}_i^{J_0} < \gamma$, then the induction hypothesis
holds trivially. Otherwise, $\tilde{p}_i^{J_1}\ge\tilde{p}_i^{J_0}\ge\gamma$,
so $\tilde{q}_i^{J_1}\ge\frac{\gamma}{2}$,
and this completes the proof.
\end{proof}

It is easy to verify that the Hessian $\nabla^2\phi$ of the dual
objective function is given by
\begin{eqnarray*}
(\nabla^2\phi(p))_{jl} &=& \frac{1}{1-\rho}\cdot
\sum_{i=1}^n b_i \cdot
\left(
        \frac { \sum_{J\in{\cal J}_i} a_{ij}^J a_{il}^J
                           ( (c_i^J)^{\rho}/(\tilde{p}_i^J)^{2-\rho} )^{\frac{1}{1-\rho}}
                }
                {   \sum_{J\in{\cal J}_i} (c_i^J/
            \tilde{p}_i^J)^{\frac{\rho}{(1-\rho)}}
                }\right. \\
&-& \left.\rho\cdot \frac  {
\left(\sum_{J\in{\cal J}_i} a_{ij}^J ((c_i^J)^{\rho} /\tilde{p}_i^J)^{\frac{1}{1-\rho}} \right)\cdot
\left(\sum_{J\in{\cal J}_i} a_{il}^J ((c_i^J)^{\rho} /\tilde{p}_i^J)^{\frac{1}{1-\rho}} \right)
         }
         {  \left( \sum_{J\in{\cal J}_i} (c_i^J
             /\tilde{p}_i^J)^{\frac{\rho}{(1-\rho)}}\right)^2
         }
\right).
\end{eqnarray*}
This allows us to prove the following bounds. Let $P$ denote the set
of price vectors that satisfy the constraints of Proposition~\ref{pr: tatonnement p}
and Claim~\ref{cl: price bounds}.
\begin{claim}\label{cl: Hessian operator norm}
There are constants $L_{\min} = L_{\min}(b,c,\rho) > 0$
and $L_{\max} = L_{\max}(b,c,\rho) > 0$ such that
for every $x\in\RR^m$ and for every $p\in P$,
$$
L_{\min}\cdot\|\tilde{x}\|_2^2 \le
\langle x \mid \nabla^2\phi(p) \mid x \rangle\le 
L_{\max}\cdot \|\tilde{x}\|_2^2.
$$
\end{claim}

\begin{proof}
We begin with the following equation:
\begin{eqnarray*}
\langle x \mid \nabla^2\phi(p) \mid x \rangle &=& \frac{1}{1-\rho}\cdot
\sum_{i=1}^n \frac{b_i}{\left( \sum_{J\in{\cal J}_i} (c_i^J
             /\tilde{p}_i^J)^{\frac{\rho}{(1-\rho)}}\right)^2} \cdot \\
&&\left(  \sum_{J\in{\cal J}_i} \sum_{J'\in{\cal J}_i}
                            (c_i^J c_i^{J'})^{\frac{\rho}{1-\rho}}
    \left(\frac{(\tilde{x}_i^J)^2}{(\tilde{p}_i^J)^{\frac{2-\rho}{1-\rho}}
        (\tilde{p}_i^{J'})^{\frac{\rho}{1-\rho}}}
-  \rho\cdot
\frac{\tilde{x}_i^J\tilde{x}_i^{J'}}{(\tilde{p}_i^J)^{\frac{1}{1-\rho}}
(\tilde{p}_i^{J'})^{\frac{1}{1-\rho}}}    \right)
\right).
\end{eqnarray*}

Fix $i$, denote $X_J = \tilde{x}_i^J$ and
$q_J = \frac{1}{\tilde{p}_i^J}$, then consider the term
$$
Z_{J,J'} = q_J^{\frac{2-\rho}{1-\rho}} q_{J'}^{\frac{\rho}{1-\rho}} X_J^2
-  \rho\cdot q_J^{\frac{1}{1-\rho}} q_{J'}^{\frac{1}{1-\rho}} X_J X_{J'}.
$$
Put $A_{J,J'} = q_J^{\frac{1-\rho/2}{1-\rho}} q_{J'}^{\frac{\rho/2}{1-\rho}}$.
Because we sum over all $J,J'\in {\cal J}_i$, we can replace
$Z_{J,J'}$ by
$$
\frac{1}{2}\cdot\left(A_{J,J'}^2 X_J^2 + A_{J',J}^2 X_{J'}^2\right) -
\rho A_{J,J'} A_{J',J} X_J X_{J'} =
$$
$$
= \frac{\rho}{2}\cdot\left(A_{J,J'} X_J + A_{J',J} X_{J'}\right)^2 +
\frac{1 - \rho}{2}\cdot\left(A_{J,J'}^2 X_J^2 + A_{J',J}^2 X_{J'}^2\right).
$$

Let
$$
L = \sum_{i=1}^n b_i\cdot
\frac{\sum_{J\in{\cal J}_i} \sum_{J'\in{\cal J}_i}
         \left((c_i^J c_i^{J'})^{\frac{\rho}{1-\rho}} / (\tilde{p}_i^J)^{\frac{2-\rho}{1-\rho}}
        (\tilde{p}_i^{J'})^{\frac{\rho}{1-\rho}}\right)\cdot
         (\tilde{x}_i^J)^2}{\sum_{J\in{\cal J}_i} \sum_{J'\in{\cal J}_i}
          (c_i^J c_i^{J'})^{\frac{\rho}{(1-\rho)}} / (\tilde{p}_i^J \tilde{p}_i^{J'})^{\frac{\rho}{(1-\rho)}}}.
$$
We have that for $\rho > 0$,
$$
L \le \langle x \mid \nabla^2\phi(p) \mid x \rangle\le\frac{1}{1-\rho}\cdot L,
$$
and for $\rho < 0$,
$$
\frac{1}{1-\rho}\cdot L \le  \langle x \mid \nabla^2\phi(p) \mid x \rangle\le L.
$$
We proceed to bound $L$.
Using Claim~\ref{cl: price bounds}, the trivial upper bound 
$\|\tilde{p}\|_\infty \le 1$
(which follows from the bounds $\|p\| _1 = 1$ and $\|a\|_1 = 1$),
the notation $b_{\min},c_{\min}$, and the fact that for all $i$,
$1\le |{\cal J}_i|\le c_{\min}^{\frac{-\rho}{1-\rho}}$, we get the
following lower and upper bounds on $L$.
If $\rho > 0$ we get that
$$
2^{\frac{\rho(\rho-2)}{(1-\rho)^2}}\cdot
b_{\min}^{\frac{1}{1-\rho}}\cdot c_{\min}^{\frac{2}{1-\rho}}\cdot \|\tilde{x}\|_2^2\le
L\le 2^{\frac{2-\rho}{(1-\rho)^2}}\cdot b_{\min}^{\frac{-2}{1-\rho}}\cdot
c_{\min}^{\frac{-(4+\rho)}{1-\rho}}\cdot \|\tilde{x}\|_2^2.
$$
Similarly, if $\rho < 0$ we get that
$$
2^{\frac{\rho(2-\rho)}{1-\rho}}\cdot b_{\min}^{1-\rho}\cdot \|\tilde{x}\|_2^2\le
L\le 2^{\frac{(3\rho-2)(\rho-2)}{1-\rho}}\cdot b_{\min}^{3\rho-2}\cdot
c_{\min}^{\frac{\rho}{1-\rho}}\cdot \|\tilde{x}\|_2^2.
$$
This completes the proof.
\end{proof}

\begin{corollary}\label{cor: Hessian operator norm}
There is a constant $\lambda_{\max} = \lambda_{\max}(a,b,c,\rho)$
such that for every $x\in\RR^m$ and for every $p\in P$,
$\|\nabla^2\phi(p)\mid x\rangle\|_2^2\le \lambda_{\max}^2 \|\tilde{x}\|_2^2$.
\end{corollary}

\begin{proof}
First notice that
$\|\nabla^2\phi(p)\mid x\rangle\|_2^2\le\max_{y\ne \vec{0}}
\frac{\langle y\mid \nabla^2\phi(p)\mid y\rangle}{\|y\|_2^2}\cdot
\langle x\mid \nabla^2\phi(p)\mid x\rangle$. By
Claim~\ref{cl: Hessian operator norm},
$\langle y\mid \nabla^2\phi(p)\mid y\rangle\le L_{\max}\cdot \|\tilde{y}\|_2^2$.
By Claim~\ref{cl: tilde norm},
$\|\tilde{y}\|_2^2\le A\cdot\|y\|_2^2$. Using Claim~\ref{cl: Hessian operator norm}
again, $\langle x\mid \nabla^2\phi(p)\mid x\rangle\le L_{\max}\cdot \|\tilde{x}\|_2^2$.
Summing up the inequalities,
$\|\nabla^2\phi(p)\mid x\rangle\|_2^2\le A\cdot L_{\max}^2\cdot \|\tilde{x}\|_2^2$.
\end{proof}

We are now ready to establish the monotonicity of the dual objective
function. 
\begin{claim}\label{cl: monotonicity}
For sufficiently small $\veps$,
the sequence $\phi(p^t)$, $t=0,1,2,\dots$, is monotonically
non-increasing.
\end{claim}

\begin{proof}
Write $p^{t+1} = p^t - q^t$, where for every $j=1,2,\dots,m$,
$q^t_j = -\veps p^t_j z_j(p^t)$. By Lemma~\ref{lm: gradient},
$q^t_j = \veps p^t_j (\nabla\phi(p^t))_j$. Consider the second
order Taylor expansion of $\phi(p^t)$ with respect to $\phi(p^{t+1})$.
$$
\phi(p^t) = \phi(p^{t+1}) + \langle q^t \mid \nabla\phi(p^{t+1})\rangle +
\frac 1 2 \langle q^t \mid \nabla^2\phi(p) \mid q^t \rangle,
$$
where $p = \gamma p^t + (1-\gamma) p^{t+1}$ for some $\gamma\in [0,1]$.
As $\phi$
is a convex function on $P$ and $p\in P$, the quadratic term in the Taylor
expansion is non-negative. Thus, our proof is complete if we show that
the linear term is also non-negative.

Write $\nabla\phi(p^{t+1}) = \nabla\phi(p^t) + (\nabla\phi(p^{t+1}) - \nabla\phi(p^t))$.
Thus,
\begin{eqnarray*}
\langle q^t \mid \nabla\phi(p^{t+1})\rangle & = &
   \langle q^t \mid \nabla\phi(p^t)\rangle +
   \langle q^t \mid \nabla\phi(p^{t+1} - \nabla\phi(p^t))\rangle \\
& = & \veps\cdot\sum_{j=1}^m p^t_j (\nabla\phi(p^t))_j^2 -
    \veps\cdot\sum_{j=1}^m p^t_j (\nabla\phi(p^t))_j \sum_{i=1}^n (x_{ij}^{t+2} - x_{ij}^{t+1}).
\end{eqnarray*}
Let $f(p)\in\RR_+^m$ denote the vector of total demands for the goods
induced by $x(p)$. In particular, $f(p^t) = \sum_{i=1}^n x_i^{t+1}$.
In order to complete the proof we show that for sufficiently small $\veps$,
$$
\sum_{j=1}^m p^t_j (\nabla\phi(p^t))_j (f_j(p^{t+1}) - f_j(p^t))\le
\sum_{j=1}^m p^t_j (\nabla\phi(p^t))_j^2.
$$

Let $p(\gamma) = (1 - \gamma) p^t + \gamma p^{t+1}$, for $\gamma\in [0,1]$.
By the fundamental theorem of line integrals,
$$
f_j(p^{t+1}) - f_j(p^t)
= \int_0^1 \langle \nabla f_j(p(\gamma))\mid p^{t+1}-p^t \rangle d\gamma
= \int_0^1 \langle -\nabla f_j(p(\gamma))\mid q^t \rangle d\gamma.
$$
Using Lemma~\ref{lm: gradient},
$-(\nabla f_j(p(\gamma)))_{j'} = (\nabla^2\phi(p(\gamma)))_{j,j'}$.
We get that
\begin{eqnarray*}
               \sum_{j=1}^m p^t_j (\nabla\phi(p^t))_j (f_j(p^{t+1}) - f_j(p^t)) 
& = & \sum_{j=1}^m p^t_j (\nabla\phi(p^t))_j 
      \int_0^1 \langle -\nabla f_j(p(\gamma))\mid q^t \rangle d\gamma \\
& = & \veps\cdot\sum_{j=1}^m p^t_j (\nabla\phi(p^t))_j \int_0^1 \sum_{j'=1}^m
         (\nabla^2\phi(p(\gamma)))_{j,j'}\cdot p^t_{j'}(\nabla\phi(p^t))_{j'} d\gamma \\
& = & \veps\cdot\int_0^1  \sum_{j=1}^m\sum_{j'=1}^m
    p^t_j (\nabla\phi(p^t))_j p^t_{j'}(\nabla\phi(p^t))_{j'} (\nabla^2\phi(p(\gamma)))_{j,j'} d\gamma. 
\end{eqnarray*}
This implies that there exists $\gamma\in [0,1]$ such that for $p = p(\gamma)$,
$$
\sum_{j=1}^m p^t_j (\nabla\phi(p^t))_j (f_j(p^{t+1}) - f_j(p^t)) \le
\veps\cdot\sum_{j=1}^m\sum_{j'=1}^m
    p^t_j (\nabla\phi(p^t))_j p^t_{j'}(\nabla\phi(p^t))_{j'} (\nabla^2\phi(p))_{j,j'}
$$
Define $x$ by putting $x_j = p^t_j(\nabla\phi(p^t))_j$, for all $j$.
By Claim~\ref{cl: Hessian operator norm}, there is
$L_{\max} = L_{\max}(b,c,\rho) > 0$ such that
for every $x$,
$\langle x\mid \nabla^2\phi(p) \mid x\rangle \le
L_{\max}\cdot \|\tilde{x}\|_2^2$.
Recall that $A = \max_j \sum_{i,J} a_{ij}^J$.
We get
\begin{eqnarray*}
\sum_{j=1}^m p^t_j (\nabla\phi(p^t))_j (f_j(p^{t+1}) - f_j(p^t)) & \le &
         \veps\cdot\langle x\mid \nabla^2\phi(p) \mid x\rangle \\
& \le & \veps\cdot L_{\max}\cdot\|\tilde{x}\|_2^2 \\
& = & \veps\cdot L_{\max}\cdot\sum_{i,J}
      \left(\sum_{j=1}^m a_{ij}^J p^t_j(\nabla\phi(p^t))_j\right)^2 \\
& \le & \veps\cdot L_{\max}\cdot\sum_{i,J}
      \sum_{j=1}^m a_{ij}^J p^t_j(\nabla\phi(p^t))_j^2 \\
& = & \veps\cdot L_{\max}\cdot\sum_{j=1}^m
        p^t_j(\nabla\phi(p^t))_j^2\cdot \sum_{i,J} a_{ij}^J \\
& \le & \veps\cdot L_{\max}\cdot A\cdot \sum_{j=1}^m
        p^t_j(\nabla\phi(p^t))_j^2,
\end{eqnarray*}
where the penultimate inequality uses Cauchy-Schwartz.
Choosing $\veps\le\frac{1}{L_{\max}\cdot A}$ completes the proof.
\end{proof}

\begin{claim}\label{cl: uniformity}
A Fisher market with nested CES-Leontief utilities is
$\beta$-uniform, for $\beta$ that satisfies
$\beta(\alpha) = \Theta(\alpha^2)$.
\end{claim}

\begin{proof}
Let $p^*\in P$ be a dual optimal solution. We show that
there exists
a constant $\kappa > 0$ such that for every $\delta > 0$ 
and for every $p\in P$, 
if $\phi(p)\le \phi(p^*) + \kappa\cdot\delta^2$,
then $|z(p) - z(p^*)|_\infty\le \delta$.

Consider the second order Taylor expansion of $\phi(p)$
with respect to $\phi(p^*)$:
$$
\phi(p) = \phi(p^*) + \langle p - p^* \mid \nabla\phi(p^*) \rangle +
  \frac 1 2 \langle p - p^* \mid \nabla^2\phi(p') \mid p - p^* \rangle.
$$
Notice that $(\nabla\phi(p^*))_j = 0$ unless $p^*_j = 0$, in which
case both $(\nabla\phi(p^*))_j\ge 0$ and $p_j - p^*_j\ge 0$, so
$\langle p - p^* \mid \nabla\phi(p^*) \rangle \ge 0$.
We get that for a constant $L_{\min} > 0$,
$$
\kappa \delta^2 \ge \phi(p) - \phi(p^*) \ge
\frac 1 2 \langle p - p^* \mid \nabla^2\phi(p') \mid p - p^* \rangle \ge
\frac 1 2\cdot L_{\min}\cdot \|\widetilde{p - p^*}\|_2^2,
$$
where the last inequality follows from Claim~\ref{cl: Hessian operator norm}.
Thus, $\|\widetilde{p - p^*}\|_2^2\le\frac{2\kappa \delta^2}{L_{\min}}$.
On the other hand, $|z_j(p) - z_j(p^*)| = |f_j(p) - f_j(p^*)|\le \|f(p) - f(p^*)\|_2$.
Consider $q(\gamma) = \gamma p + (1-\gamma) p^*$. By the fundamental
theorem of line integrals,
$f_j(p) - f_j(p^*) = \int_0^1 \langle \nabla(f_j(q(\gamma)))\mid p - p^*\rangle d\gamma$.
Therefore,
$$
\|f(p) - f(p^*)\|_2^2 = \sum_j \left(
    \int_0^1 \langle \nabla(f_j(q(\gamma)))\mid p - p^*\rangle d\gamma\right)^2
\le \sum_j \int_0^1\left(\langle \nabla(f_j(q(\gamma)))\mid p - p^*\rangle\right)^2 d\gamma.
$$
(The inequality is a Cauchy-Schwartz argument---consider the random variable
$\langle \nabla(f_j(q(\gamma)))\mid p - p^*\rangle$ on $\gamma\in [0,1]$ endowed
with the uniform probability measure.)
Therefore, there exists $\gamma\in [0,1]$ such that for $q = q(\gamma)$ we
have
$$
\|f(p) - f(p^*)\|_2^2\le \sum_j \left(\langle \nabla(f_j(q))\mid p - p^*\rangle\right)^2
= \|\nabla^2\phi(q) \mid p - p^*\rangle\|_2^2
\le \lambda_{\max}^2\cdot \|\widetilde{p - p^*}\|_2^2\le
\frac{2\kappa \delta^2\lambda_{\max}^2}{L_{\min}},
$$
where the last inequality follows from Corollary~\ref{cor: Hessian operator norm}
and $\lambda_{\max}$ is the constant stipulated by that corollary.
To complete the proof, choose $\kappa = \frac{L_{\min}}{2\lambda_{\max}^2}$.
\end{proof}

We are now ready for the proof of convergence of
t\^{a}tonnement in Fisher markets with nested CES-Leontief
utilities.
\begin{theorem}\label{thm: CES-Leontief}
There are constants $\kappa_1 = \kappa_1(a,b,c,\rho)$
and $\kappa_2 = \kappa_2(a,b,c,\rho)$ such that the
following holds. If $\veps\le \frac{1}{\kappa_1}$, then for 
every $\delta > 0$ and for every
$T\ge \frac{\kappa_2\ln(1/p^0_{\min})}{\veps^2\delta^3}$,
the price-demand pair
$(p^T,x^{T+1})$ is a $\delta$-approximate equilibrium
in the sense of Definition~\ref{def: approx-equilib 1}.
\end{theorem}

\begin{proof}
By Claim~\ref{cl: uniformity} and Claim~\ref{cl: monotonicity},
the conditions stated in Lemma~\ref{lm: better mwum} are
satisfied, so the conclusion of the lemma holds. The two
claims and the lemma together establish the conditions
stated in Lemma~\ref{lm: main}, which in turn shows the
convergence claim and rate stated in the theorem.
\end{proof}

\section{Additional Results}\label{sec: additional}

A market with resource allocation utilities is similar to
a market with nested CES-Leontief utilities, except
that $\rho$ is set to $1$ in the case of resource allocation
utilities. In other words, the utility functions are given by
$u_i(x) = \sum_{J\in{\cal J}_i} c_i^J \min_{j\in J_i}
\left\{\frac{x_{ij}^J}{a_{ij}^J}\right\}$.
Resource allocation markets  generalize both
Leontief utilities and linear utilities. The
t\^{a}tonnement process is not known to converge
in the case of linear utilities (and in fact is unlikely
to converge in that case), so we need to apply the
process to distorted utilities (see~\cite{CF08} for
the case of linear utilities).
Notice that the reactions of the buyers are
assumed to be optimal with respect to the distorted utilities
and not the original utilities. (In the proportional response
dynamics that converge to equilibrium in the case of linear
utilities~\cite{BDX11} the agents also do not respond
optimally to the prices.)

We will replace the utility function of each agent
by a nested CES-Leontif utility. This is detailed in
the proof of the following theorem that analyzes
the distorted utilities process.
\begin{theorem}\label{thm: resource allocation}
Let $k = \max_i |{\cal J}_i|$.
For every $\delta > 0$ there are constants $\kappa_0 = \kappa_0(c)$,
$\kappa_1 = \kappa_1(a,b)$,
and $\kappa_2 = \kappa_2(a,b,c)$ such that the following holds.
For $\veps\le \frac{1}{\kappa_0^{\log^2 k/\delta^2}\kappa_1^{\log k/\delta}}$,
and $T\ge \frac{\kappa_2\ln(1/p^0_{\min})}{\veps^2\delta^3}$,
the price-demand pair
$(p^T,x^{T+1})$ is a $\delta$-approximate equilibrium
in the sense of Definition~\ref{def: approx-equilib 2}.
\end{theorem}

\begin{proof}
We replace the utility functions by their distorted versions
$\tilde{u}_i(x) = \left(\sum_{J\in{\cal J}_i} \left(c_i^J \min_{j\in J_i}
\left\{\frac{x_{ij}^J}{a_{ij}^J}\right\}\right)^\rho\right)^{1/\rho}$,
for $\rho = 1 - \frac{\delta}{4\ln k}$. We apply
Theorem~\ref{thm: CES-Leontief} to get prices $p$ and
allocations $x$ that are a $\frac{\delta}{2}$-approximate
equilibrium in the sense of Definition~\ref{def: approx-equilib 1}
for the utilities $\tilde{u}_i$. Notice that by properties P2 and
P3 of Definition~\ref{def: approx-equilib 1}, the allocations
$\tilde{x} = \frac{1}{1+\delta/2} x$ satisfy properties P2 and P3
of Definition~\ref{def: approx-equilib 2}. Moreover, as the
distorted utility functions are $1$-homogeneous (this is
also true of resource allocation utilities),
$\tilde{u}_i(\tilde{x}) = \frac{1}{1+\delta/2} \tilde{u}_i(x)$.
To complete the proof, denote by $x^*$ the optimal
allocations with respect to the prices $p$ for the original
resource allocation utilities. We have that for every $i=1,2,\dots,n$,
$u_i(\tilde{x})\ge k^{1 - 1/\rho}\cdot \tilde{u}_i(\tilde{x}) \ge
(1 - \delta/2)\cdot \tilde{u}_i(\tilde{x}) = \frac{1-\delta/2}{1+\delta/2}\cdot
\tilde{u}_i(x)\ge \frac{1-\delta/2}{1+\delta/2}\cdot\tilde{u}_i(x^*)\ge
\frac{1-\delta/2}{1+\delta/2}\cdot u_i(x^*)\ge (1- \delta)\cdot u_i(x^*)$.
This establishes property P1 of Definition~\ref{def: approx-equilib 2}.
\end{proof}

\bibliographystyle{plain}
%
%
%
\bibliography{references}

\newpage
\appendix

\section*{Appendix: Proofs}

\begin{proofof}{Lemma~\ref{lm: gradient}}
By the definition of $g_i^*$,
$$
\phi(p) = \max_{x\in\RR_{++}^{n\times m}} \left\{\sum_{i=1}^n b_i \ln u_i(x_i) +
    \sum_{j=1}^m p_j \left(1 - \sum_{i=1}^n x_{ij}\right)\right\}.
$$
Let $\phi_x(q) = \sum_{i=1}^n b_i \ln u_i(x_i) +
\sum_{j=1}^m q_j \left(1 - \sum_{i=1}^n x_{ij}\right)$. This
is a linear function of $q$. Notice that
$\phi(q) = \max_x \phi_x(q)$. Fix $x$ to be a maximizing
assignment for $q = p$. By a well-known fact,
$\nabla\phi(p) = \nabla\phi_x(p)$. But $\phi_x(p)$
is a linear function of $p$ and its gradient is given
by $(\nabla\phi_x(p))_j = 1 - \sum_{i=1}^n x_{ij}$.

Now, to complete the proof, we show that the maximizing
assignment $x$ for $q=p$ is $x = x(p)$.
Notice that
$$
\arg\max_{x\in\RR_{++}^{n\times m}} \left\{\sum_{i=1}^n b_i \ln u_i(x_i) +
    \sum_{j=1}^m p_j \left(1 - \sum_{i=1}^n x_{ij}\right)\right\} =
$$
$$
\arg\max_{x\in\RR_{++}^{n\times m}} \left\{\sum_{i=1}^n b_i \ln u_i(x_i) +
    \sum_{j=1}^m p_j \left(\sum_{i=1}^n x_{ij}(p) - \sum_{i=1}^n x_{ij}\right)\right\},
$$
because the expressions on both sides of the equation differ only by an additive constant
that does not depend on $x$. Finally, notice that the solution to the right-hand
side is an equilibrium demand for the same agents, but where the supply of
each good $j$ is equal to $\sum_{i=1}^n x_{ij}(p)$. The equilibrium demand
in this market is precisely $x(p)$.
\end{proofof}

\bigskip
\begin{proofof}{Lemma~\ref{lm: mwum}}
The proof uses the standard potential function argument
analyzing the multiplicative weights update method. We
begin with the following inequality that holds for every
$t=0,1,2,\dots$.
\begin{eqnarray*}
\sum_{j=1}^m p^{t+1}_j & = &\sum_{j=1}^m p^t_j \left(1 + \veps z_j(p^t)\right) \\
& = & \sum_{j=1}^m p^t_j +
\veps\cdot\left(\sum_{j=1}^m p^t_j\right)\cdot\sum_{j=1}^m \frac{p^t_j}{\sum_{j'=1}^m p^t_{j'}}z_j(p^t) \\
& = & \left(\sum_{j=1}^m p^t_j\right)\cdot
    \left(1 + \veps\sum_{j=1}^m \frac{p^t_j}{\sum_{j'=1}^m p^t_{j'}}z_j(p^t)\right) \\
& \le & \left(\sum_{j=1}^m p^t_j\right)\cdot e^{\veps\sum_{j=1}^m p^t_j z_j(p^t) / \sum_{j'=1}^m p^t_{j'}}.
\end{eqnarray*}
So, on the one hand,
$$
\sum_{j=1}^m p^T_j \le \left(\sum_{j=1}^m p^0_j\right)\cdot
    e^{\veps\sum_{t=0}^{T-1}\sum_{j=1}^m p^t_j z_j(p^t) / \sum_{j'=1}^m p^t_{j'}}.
$$
On the other hand, for every $k\in\{1,2,\dots,m\}$,
$$
\sum_{j=1}^m p^T_j\ge  p^T_k = p^0_k\cdot\prod_{t=0}^{T-1}\left(1 + \veps z_k(p^t)\right).
$$
Taking the logarithms of the lower and upper bounds for $\sum_{j=1}^m p^T_j$,
we get that
$$
\sum_{t=0}^{T-1} \ln\left(1 + \veps z_k(p^t)\right)\le
\ln\left(\frac{\sum_{j=0}^m p^0_j}{p^0_k}\right) +
\veps\cdot\sum_{t=0}^{T-1}\sum_{j=1}^m \frac{p^t_j}{\sum_{j'=1}^m p^t_{j'}} z_j(p^t).
$$
The second term on the right-hand side
equals $0$ (see the proof of Lemma~\ref{lm: mwum response}
and recall that $z_j(p^t) = 1 - \sum_{i=1}^n x(p^t)_{ij}$).
Using the fact that $\ln(1+\xi)\ge \xi - \xi^2$ for every
$\xi\in \left[-\frac 1 2,+\frac 1 2\right]$, we get that
$$
\sum_{t=0}^{T-1} \veps z_k(p^t) - \sum_{t=0}^{T-1} \veps^2 (z_k(p^t))^2\le
  \ln\left(\frac{\sum_{j=1}^m p^0_j}{p^0_k}\right).
$$
Averaging over $t$, we get that
$$
\frac{1}{T}\cdot\sum_{t=0}^{T-1} z_k(p^t)\le
\frac{1}{T}\cdot\frac{\ln(1/p^0_k)}{\veps} +
\frac{1}{T}\cdot\sum_{t=0}^{T-1} \veps\cdot (z_k(p^t))^2\le
\frac{\ln(1/p^0_k)}{\veps T} + \veps v.
$$
(Recall that we scale $b$ so that $\sum_{j=1}^m p^0_j = \sum_{i=1}^n
b_i = 1$.)
\end{proofof}

\end{document}